%% file: template.tex
\documentclass[a4paper]{article}

\usepackage{arxiv}

\usepackage[utf8]{inputenc} 
\usepackage[T1]{fontenc}    
\usepackage{hyperref}       
\usepackage{url}            
\usepackage{booktabs}       
\usepackage{amsfonts}       
\usepackage{nicefrac}       
\usepackage{microtype}      
\usepackage{lipsum}		
\usepackage{graphicx}
\usepackage{natbib}
\usepackage{doi}
\usepackage{latexsym}
\usepackage{amsmath}
\usepackage{amssymb}
\usepackage[mathscr]{eucal}
\usepackage{graphicx,color}
\usepackage{mathtools}
\usepackage{stmaryrd}
\usepackage{url}
\usepackage{xypic}

\usepackage{bm}
\input{notation_defs}

\usepackage{bbm}
\usepackage{tikz-cd}
\usepackage{psfrag}

\def \R {{\mathbb R}}

\def \S {{\mathbb S}}
\def \L {{\mathbb L}}

\providecommand{\ddt}{\frac{d}{dt}}

\usepackage{amsthm}
\newtheorem{theorem}{Theorem}
\newtheorem{proposition}{Proposition}
\newtheorem{definition}{Definition}
\newtheorem{assumption}{Assumption}
\newtheorem{lemma}{Lemma}
\newtheorem{remark}{Remark}

\title{
Equivariant Observer Design on $\SL(3)$ for Image Intensity-Based Homography Estimation}

    \author{ \href{https://orcid.org/0000-0003-1116-7415}{\includegraphics[scale=0.06]{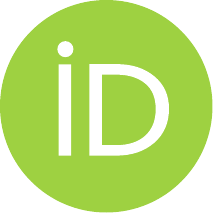}\hspace{1mm}Tarek Bouazza} \\
	I3S, CNRS, Université Côte d'Azur\\
    06900 Sophia Antipolis, France \\
	\texttt{bouazza@i3s.unice.fr} \\
	\And
    \href{https://orcid.org/0000-0003-4391-7014}{\includegraphics[scale=0.06]{orcid.pdf}\hspace{1mm}Pieter van Goor} \\
	School of Aerospace, Mechanical, and \\ Mechatronic Engineering (AMME) \\
        Faculty of Engineering, University of Sydney \\
        NSW, 2006, Australia \\
	\texttt{pieter.vangoor@sydney.edu.au} \\
	\And
    \href{https://orcid.org/0000-0002-7803-2868}{\includegraphics[scale=0.06]{orcid.pdf}\hspace{1mm}Robert Mahony} \\
	Systems Theory and Robotics Group\\
        Australian National University\\
	ACT, 2601, Australia \\
	\texttt{Robert.Mahony@anu.edu.au} \\
    \And
     \href{https://orcid.org/0000-0002-7779-1264}{\includegraphics[scale=0.06]{orcid.pdf}\hspace{1mm}Tarek Hamel} \\
        I3S, CNRS, Université Côte d'Azur\\
        and Institut Universitaire de France (IUF) \\
        06900 Sophia Antipolis, France \\
        \texttt{thamel@i3s.unice.fr} \\
}

\renewcommand{\headeright}{Preprint}
\renewcommand{\undertitle}{Preprint\thanks{Under review}}

\renewcommand{\shorttitle}{}

\hypersetup{
pdftitle={A template for the arxiv style},
pdfsubject={q-bio.NC, q-bio.QM},
pdfauthor={David S.~Hippocampus, Elias D.~Striatum},
pdfkeywords={First keyword, Second keyword, More},
}

\begin{document}
\maketitle

\begin{abstract}
This paper addresses the problem of homography estimation using a nonlinear observer designed on the Lie group $\SL(3)$ that exploits the full image information through direct image registration.
Unlike traditional feature-based methods, which rely on extensive feature extraction and matching, the proposed approach formulates an observer that minimises a cost function defined directly in terms of image pixel intensities.
Explicit conditions ensuring the non-degeneracy of the cost function are derived, and a comprehensive analysis is conducted to characterise and generate degenerate (unobservable) image configurations.
Theoretical results demonstrate local exponential convergence of the observer.
To improve local convergence properties, a second-order observer variant is introduced by incorporating the Hessian of the cost function into the correction term.
Simulation results demonstrate the performance of the proposed solutions on real images.
\end{abstract}

\section{Introduction}

Advances in exteroceptive sensors, such as high-resolution cameras, LiDAR and RGB-D devices, have enabled autonomous systems to capture detailed information about their environment. Exploiting this information is critical for applications such as Simultaneous Localisation and Mapping (SLAM) \citep{forster2014svo} and 3D reconstruction \citep{geiger2011stereoscan}.
    Directly estimating rigid-body motion and geometric transformations from raw measurements presents a promising alternative to traditional landmark- or feature-based methods, which rely on extracting and matching sparse keypoints and can fail in environments where features are scarce or affected by noise.
    Dense methods exploit the full information available across the sensor domain, such as pixel intensities, optical flow fields, or dense depth maps to achieve more robust and accurate estimation \citep{forster2014svo,kerl2013robust}.
	A prominent example is the computation of the \emph{homography} matrix, which describes the relative transformation between two camera views of a planar surface \citep{hartley2003multiple}. 
    
	Homography estimation has been widely studied in the computer vision, computer graphics, and robotics communities. 
    Traditional techniques extract a sparse set of salient features, such as points, lines, conics and contours, to perform feature matching and then compute a homography estimate by solving algebraic constraints or optimisation problems over individual frames or a set of selected images \citep{kaminski2004multiple,agarwal2005survey}. 
	One widely used method is the Direct Linear Transform (DLT) algorithm \citep{hartley2003multiple}, which computes the homography from at least four point correspondences by solving a linear system via Singular Value Decomposition (SVD).
	Direct (or intensity-based) methods, on the other hand, estimate homographies using raw image pixel intensities without explicit feature extraction or matching \citep{szeliski2007image}. They generally outperform feature-based methods in situations where feature detection is difficult or images lack distinctive features.
    However, direct methods typically have a limited convergence range and require complex optimisation techniques and higher computational cost to efficiently process the entire image data, especially for high-resolution images.
	Classical direct approaches formulate the problem as an image registration (or alignment) task \citep{shum1998construction}, which aims to find the homography parameters that minimise a photometric cost between a reference and a target image, and optimise a cost function based on pixel intensity differences between images. Common cost functions include the sum of squared differences \citep{baker2004lucas,benhimane2004real} and the normalised cross-correlation \citep{szeliski2007image}. The resulting optimisation problems are solved iteratively using algorithms such as steepest-descent, Gauss-Newton, and Levenberg-Marquardt \citep{baker2004lucas,szeliski1995direct}. 
	A seminal method is the Lucas-Kanade algorithm \citep{lucas1981iterative}, which updates the homography parameters iteratively to minimize intensity differences. The forward compositional algorithm \citep{shum2002construction} improves this approach by estimating an incremental warp rather than an additive increment, while the inverse compositional method \citep{baker2004lucas} further improves efficiency by reversing the roles of reference and target images, allowing pre-computation of the inverse Hessian and steepest descent images. 
	The efficient second-order minimization (ESM) method \citep{benhimane2004real} approximates the Hessian to achieve faster convergence and a larger convergence domain.

    In recent years, significant work has been dedicated to designing nonlinear observers for the estimation of homographies with dynamics.
	A number of successful approaches have exploited the $\SL(3)$ Lie group structure of the set of homographies \citep{benhimane2007homography} to yield powerful stability guarantees \citep{mahony2012nonlinear,hamel2011homography,hua2019feature,hua2020nonlinear}.
    By exploiting the temporal correlations and velocity information inherent to robotics problems, these methods outperform algorithms that process each image independently.
	\cite{hamel2011homography} proposed a nonlinear deterministic observer on $\SL(3)$ using point-feature correspondences built upon Lyapunov-based design principles. 
	Later, \cite{hua2019feature} extended this line of work by developing a recursive observer from the general theory of gradient-based nonlinear observers for kinematic systems with equivariant outputs \citep{mahony2013observers}.
	Further extensions of the framework in \citep{mahony2013observers} incorporated richer feature modalities.
    For instance, image line and point-feature correspondences were combined in \citep{hua2020nonlinear} for computing the innovation term, while \cite{hua2017explicit} considered conic-feature correspondences such as ellipses and hyperbolas. 
	Beyond deterministic designs, \cite{bernal2023bayesian} formulated a Bayesian filtering framework on $\SL(3)$, and derived both an iterative EKF and an interacting multiple model (IMM) filter to estimate the homography using point-feature correspondences. 
	More recently, an equivariant filter (EqF) was proposed by the authors in \citep{bouazza2023equivariant} to estimate homography from rigid-body velocity measurements and point-feature correspondences.

This paper build on the authors’ prior work on intensity-based homography estimation on $\SL(3)$ \citep{bouazza2023nonlinear}
and proposes a new class of nonlinear observers that directly exploit the photometric information contained in image data.
The approach is based on the natural group action of $\SL(3)$ on the sphere $\S^2$ and its induced action on the Sobolev space of weakly twice differentiable square-integrable functions $W^{2,2}(\S^2)$.
Images are modelled as functions warped by an evolving homography, and a photometric cost function is defined through this group action. A nonlinear observer on $\SL(3)$ is then designed to minimise this cost and estimate the dynamic homography.
The main contributions are:
\begin{enumerate}
	\item We show that the action of $\SL(3)$ on the sphere $\S^2$ induces an action of $\SL(3)$ on $W^{2,2}(\S^2)$ and exploit this to define a direct photometric cost function for homography estimation from raw image intensities. 
	We then derive a gradient-based observer on $\SL(3)$ to minimise this cost.

    \item We perform an observability analysis and identify a local non-degeneracy condition characterised by the stabilizer subgroup of the reference image under $\SL(3)$, and illustrate degenerate cases with unobservable image configurations.

    \item We prove global stability and local exponential convergence of the observer under the derived non-degeneracy condition and introduce a variant that scales the correction by the inverse Hessian of the cost function to improve convergence. 

    \item We present simulation results on real images to validate the convergence of the proposed observers. 
\end{enumerate}

The remainder of the paper is organised as follows. Section \ref{sec:notation} introduces the necessary preliminaries on square-integrable functions on the sphere and the special linear group $\SL(3)$. Section \ref{sec:observability} formulates the homography estimation problem and presents a detailed observability analysis in the dense image intensity measurement setting, where observability conditions are established. Section \ref{sec:observer_derivation} derives the dense observer on $\SL(3)$, including the gradient-based and second-order variants to improve convergence. 
Section \ref{sec:simulation} presents simulation results on real image sequences to demonstrate the effectiveness and robustness of the proposed observer.
Finally, Section \ref{sec:conclusion} concludes the paper and discusses directions for future research.

\section{Preliminary material}
\label{sec:notation}

	\subsection{Notation}

    We denote by $\R$ and $\R_{\geq0}$ the sets of real and nonnegative real numbers, respectively. The $n$-dimensional Euclidean space is denoted by $\R^n$.
	We denote by $\R^{m \times n}$ the set of real $m \times n$ matrices. The set of $n \times n$ symmetric positive definite matrices is denoted by $\S_+(n)$, and the identity matrix is denoted by $I_n \in \R^{n\times n}$. 
	For any matrix $K \in \S_+(n)$, we define the $K$-weighted inner product of two vectors $x,y \in \R^n$ by $\langle x , y \rangle_K = x^\top K y$ and the associated norm as
	$|x|_K := \sqrt{x^\top K x}$.
	For $K = I_n$, they reduce to the standard Euclidean inner product and norm,
	\begin{align}
    \langle x, y \rangle &:= x^\top y  & |x| := \sqrt{x^\top x}.
	\end{align}
	Given $a \in \R^3$, let $a^\times$ denote the skew-symmetric matrix associated with the vector (cross) product,
	which satisfies $a^\times b = a \times b$ for all $ a, b \in \R^3$. 
    
For matrices $A, B \in \R^{n \times n}$, the Euclidean matrix inner product and associated Frobenius norm are defined as
\begin{equation}
	\langle A, B \rangle := \tr(A^\top B),
	\qquad
	|A| := \sqrt{\langle A, A \rangle}.
	\label{eq:matrix_inner_product}
\end{equation}
Any square matrix $A \in \R^{3 \times 3}$ admits the decomposition
$A = \mathbb{P}_{\mathrm{s}}(A) + \mathbb{P}_{\mathrm{a}}(A)$,
where the symmetric and skew-symmetric components are defined by
$$
\mathbb{P}_{\mathrm{s}}(A) := \tfrac{1}{2}(A + A^\top), \quad
\mathbb{P}_{\mathrm{a}}(A) := \tfrac{1}{2}(A - A^\top).
$$

\subsection{Square-integrable functions on embedded manifolds}
Given a smooth manifold $\calM$, $\tT_{\mathbf{x}}\calM$ denotes the tangent space at $\mathbf{x} \in \calM$, and $\mathfrak{X}(\calM)$ the set of smooth vector fields on $\calM$.
Given $f : \calM \rightarrow \calN$ and $g : \calN \rightarrow \calN'$, their composition is denoted by $g \circ f : \calM \rightarrow \calN'$, where $(g \circ f)(\mathbf{x}) = g(f(\mathbf{x}))$.

Let $\mathbb{S}^2 := \{\mathbf{x} \in \R^{3} \mid \vert \mathbf{x} \vert = 1\} \subset \mathbb{R}^3$ be the 2-dimensional unit sphere with the Riemannian metric induced by the Euclidean inner product in $\mathbb{R}^3$. For any point $\mathbf{x} \in \mathbb{S}^2$, the tangent space is $\tT_{\mathbf{x}}\S^2 := \{ u \in \R^3 \mid \langle u, \mathbf{x}\rangle = 0 \}$.
A convenient choice of local coordinates on $\S^2$ is given by the spherical parametrization $(\theta,\phi) \in [0,\pi] \times [0,2\pi)$, where $ \mathbf{x}(\theta,\phi) = ( \sin\theta \cos\phi, \, \sin\theta \sin\phi, \,  \cos\theta )^\top$.
In these coordinates, the associated area element (Riemannian volume form) is $dV_g = \sin\theta \, d\theta \, d\phi$,
which is the standard surface area element on $\S^2$. 
For any $\mathbf{x} \in \S^2$, the orthogonal projection onto $\tT_{\mathbf{x}}\mathbb{S}^2$ is defined by
$$\pi_{\mathbf{x}} := I_3 - \mathbf{x} \mathbf{x}^\top. $$

A square-integrable function on $\S^2$ is a real-valued map $f: \S^2 \to \R$ such that
\begin{align*} 
	\Vert f \Vert^2 := \int_{\S^2} \vert f(\mathbf{x}) \vert^2 dV_g
\end{align*}
exists and is finite.
The space of all square-integrable functions on $\S^2$ is denoted $\L^2(\S^2)$, and is a Hilbert space; that is, a real vector space equipped with an inner product defined by
\begin{align*}
	\langle f, h \rangle := \int_{\S^2} f(\mathbf{x}) h(\mathbf{x})  dV_g, \quad \forall f,h \in \L^2(\S^2).
\end{align*}

We denote by $\nabla$ the Levi-Civita connection on $\S^2$.
The covariant derivative of a tangent vector field $X \in \mathfrak{X}(\S^2)$ along $u \in \tT_{\mathbf{x}}\S^2$ is defined by projecting the usual Euclidean derivative to the tangent space $\nabla_u X := \pi_{\mathbf{x}} \big(\tD \tilde{X}(\mathbf{x})[u]\big)$,
where $\tilde{X}$ is any smooth extension of $X$ to a neighborhood of $\mathbf{x}$ in $\R^3$.

For a smooth function $f:\S^2 \to \R$, the covariant derivative along $u \in \tT_{\mathbf{x}}\S^2$ coincides with the directional derivative
$$
\nabla_u f = \tD f(\mathbf{x})[u] = \frac{d}{dt}\Big|_{t=0} f(\gamma(t)),
$$
for any smooth curve $\gamma$ with $\gamma(0) = \mathbf{x}$, $\dot{\gamma}(0) = u$.  
The gradient $\nabla f$ satisfies
$$
\nabla_u f = \langle \nabla f(\mathbf{x}), u \rangle, \quad \forall u \in \tT_{\mathbf{x}}\S^2,
$$
and is given by $\nabla f(\mathbf{x}) = \pi_{\mathbf{x}} \nabla_{\R^3} f(\mathbf{x})$.  
The Hessian is
$$
\nabla^2 f(\mathbf{x})[u,v] := \langle \nabla_u \nabla f, v \rangle, \quad \forall u,v \in \tT_{\mathbf{x}}\S^2.
$$

\begin{definition}[Weak derivative on $\S^2$, \citep{chan2024meyers}] 
A function $f: \S^2 \rightarrow \R$ is said to admit a \emph{weak $k$-th order covariant derivative}, $k \geq 1$, if there exists a covariant $k$-tensor field $\tilde{\nabla}^k f$ such that
$$ 
	\int_{\S^2} f(\mathbf{x}) \nabla^k \phi(\mathbf{x}) d V_g = (-1)^k \int_{\S^2} \left\langle \tilde{\nabla}^k f(\mathbf{x}), \phi(\mathbf{x}) \right\rangle d V_g
$$
for all smooth, compactly supported $k$-tensor fields $\phi$ on $\mathbb{S}^2$.
In particular, $\tilde{\nabla} f$ denotes the weak gradient and $\tilde{\nabla}^2 f$ is the weak Hessian of $f$. 
\end{definition}

\begin{definition}[Sobolev space, \citep{hebey1996sobolev}]
The Sobolev space $W^{k,2}(\S^2)$ is defined as  
\begin{align*}
	W^{k,2}(\S^2) := \{ f : \S^2 \to \R \mid f, \tilde{\nabla}^l f \in \L^2(\S^2), \forall l = 1,\dots,k   \}.
\end{align*}
This space consists of all functions in $\L^2(\S^2)$ whose weak derivatives up to order $k$ are also square-integrable.
\end{definition}

When $f$ is smooth on an open set in $\S^2$, the weak derivatives coincide with the classical ones almost everywhere on this set, i.e., $\tilde{\nabla} f = \nabla f$, $\tilde{\nabla}^2 f = \nabla^2 f$. 
The weak derivatives, however, are still well-defined even if $f$ fails to be differentiable on a set of measure-zero in $\S^2$ (e.g., at isolated points or along curves). 

\subsection{Special Linear Group $\SL(3)$}

For a detailed introduction to matrix Lie groups, the reader is referred to \citep{baker2003matrix}.

A \emph{matrix Lie group} $\grpG$ is a subgroup of the general linear group $\GL(n) := \{ A \in \mathbb{R}^{n \times n} \mid \det(A) \neq 0 \}$ closed under matrix multiplication and inversion.
Its Lie algebra is the linear space identified with the tangent space at the identity,
$\mathfrak{g} := \tT_{I_n}\grpG \subset \mathbb{R}^{n\times n}$.
The exponential map $\exp: \mathfrak{g} \to \grpG$ defines a local diffeomorphism from a neighbourhood of $0 \in \mathfrak{g}$ to a neighbourhood of $I_n \in \grpG$.

The inner product \eqref{eq:matrix_inner_product} naturally induces a right-invariant Riemannian metric on $\grpG$ via
$$
\langle U_1 X, U_2 X \rangle_X := \langle U_1, U_2 \rangle,
\qquad U_1, U_2 \in \gothg,\; X \in \grpG.
$$
The special orthogonal group $\SO(3)$ of three-dimensional rotations and its Lie algebra $\so(3)$ are defined by
\begin{align*}
	\SO(3) &:= \{ R \in \R^{3 \times 3} \mid \det(R) = 1, RR^\top = R^\top R = I_3 \}, \\
	\so(3) &:= \{\Omega^\times \in \R^{3\times 3} \mid \Omega \in \R^3\}.
\end{align*}
The special linear group $\SL(3)$ of matrices with unit determinant and its Lie algebra $\gothsl(3)$ are defined by
\begin{align*}
	\SL(3) &:= \{ H \in \R^{3 \times 3} \mid \det(H) = 1 \}, \\
	\gothsl(3) &:= \{ U \in \R^{3\times 3} \mid \text{tr}(U) = 0 \}.
\end{align*}
The projection $\pr_{\SL(3)}: \R^{3 \times 3} \rightarrow \SL(3)$ for any non-singular matrix $H \in \R^{3 \times 3}$ is defined as
\begin{equation} \label{eq:SL3_proj}
	\pr_{\SL(3)}(H) := \det(H)^{-\frac{1}{3}}H \in \SL(3).
\end{equation} 
Let $\{e_1,e_2,e_3\}$ denote the canonical basis of $\mathbb{R}^3$, and let $\{B_j\}_{j=1}^8$ be the following orthonormal basis of $\mathfrak{sl}(3)$ with respect to the Frobenius inner product:
\begin{align*}
B_{1} &= \tfrac{1}{\sqrt{2}} \big( e_{1}e_{1}^{\top} - e_{2}e_{2}^{\top} \big), &
B_{2} &= \tfrac{1}{\sqrt{2}} \big( e_{1}e_{2}^{\top} + e_{2}e_{1}^{\top} \big), \\
B_{3} &= \tfrac{1}{\sqrt{2}} \big( e_{1}e_{3}^{\top} + e_{3}e_{1}^{\top} \big), &
B_{4} &= \tfrac{1}{\sqrt{2}} \big( e_{2}e_{3}^{\top} + e_{3}e_{2}^{\top} \big), \\
B_{5} &= \tfrac{1}{\sqrt{2}} \big( e_{1}e_{2}^{\top} - e_{2}e_{1}^{\top} \big), &
B_{6} &= \tfrac{1}{\sqrt{2}} \big( e_{1}e_{3}^{\top} - e_{3}e_{1}^{\top} \big), \\
B_{7} &= \tfrac{1}{\sqrt{2}} \big( e_{2}e_{3}^{\top} - e_{3}e_{2}^{\top} \big), &
B_{8} &= \tfrac{1}{\sqrt{6}} \big( e_{1}e_{1}^{\top} + e_{2}e_{2}^{\top} - 2e_{3}e_{3}^{\top} \big).
\end{align*}
The \emph{wedge} operator $(\cdot)^\wedge : \mathbb{R}^8 \to \mathfrak{sl}(3)$ is the linear isomorphism defined by
$v^\wedge := \sum_{j=1}^8 v_j B_j$, for $v \in \mathbb{R}^8$.
The \emph{vee} $(\cdot)^\vee : \mathfrak{sl}(3) \to \mathbb{R}^8$ denotes its inverse, given by $V^\vee := \big( \langle V, B_1 \rangle, \dots, \langle V, B_8 \rangle \big)^\top$, $V \in \mathfrak{sl}(3)$.

	A right action of $\SL(3)$ on the sphere $\S^2$ is a smooth map $\rho : \SL(3) \times \S^2 \rightarrow \S^2$ that satisfies the \textit{identity} and \textit{compatibility} properties:
    \begin{align*}
        \rho(I_3, \mathbf{y}) &= \mathbf{y}, &
        \rho\left(G, \rho(H, \mathbf{y}) \right) &= \rho\left( H G, \mathbf{y} \right),
    \end{align*}
    for any $H, G\in \SL(3)$, $\mathbf{y} \in \S^2$. 
	The map $\rho_H : \S^2 \rightarrow \S^2$, for any fixed $H \in \SL(3)$, is a diffeomorphism $\rho_H(\mathbf{y}) := \rho(H, \mathbf{y})$. Likewise, for any fixed $\mathbf{y} \in \S^2$, $\rho_{\mathbf{y}}: \SL(3) \rightarrow \S^2$ defines a smooth projection $\rho_{\mathbf{y}}(H) := \rho(H, \mathbf{y})$.

	\subsection{ Homographies as elements of the special linear group $\SL(3)$}
	
	Consider a camera moving in space while observing a stationary planar scene. As it moves, the camera captures images of the scene from different viewpoints. 
    Let $\{ \mathring{\calC} \}$ denote the reference camera frame, associated with the viewpoint at a fixed initial time, and let $\{ \calC_t \}$ denote the current camera frame at time $t$.
	
    Let $R \in \SO(3)$ and $\xi \in \R^3$ denote the orientation and position, respectively, of the frame $\{ \mathcal{C}_t \}$ with respect to the reference frame $\{ \mathring{\mathcal{C}} \}$.
    Let $d$ denote the distance from the origin of $\{ \mathcal{C}_t \}$ to the planar scene and $\eta \in \mathbb{S}^2$  the normal vector pointing to the scene expressed in $\{ \mathcal{C}_t \}$.
    The plane is defined by 
	$\mathbf{\Pi} := \{  P \in \R^3 \mid \eta^\top P - d = 0 \}$.
	Let $\mathring{P} \in \{ \mathring{\mathcal{C}} \}$ and $P \in \{ \mathcal{C}_t \}$ be the 3D coordinate vectors of a point belonging to the scene, related by
	\begin{equation} \label{eqn:mp}
		\mathring{P} = R P + \xi.
	\end{equation}
Under the pinhole model \citep{ma2004invitation}, the intrinsic camera parameters, such as the focal length and the principal point, are captured by a calibration matrix $K \in \R^{3\times 3}$, such that the projection of a point $\mathring{P}$ (resp. $P$) onto the image plane is $\mathring{\zeta} = K \mathring{P}$ (resp. $\zeta = K P$).
The image homography $H_{im} \in \R^{3 \times 3}$ that maps pixel coordinates from $\{\calC_t\}$ to $\{\mathring{\calC}\}$ is given by
$$ 
H_{im} := \gamma K \left( R + \frac{\xi \eta^\top}{d} \right) K^{-1},
$$
where $\gamma > 0$ is an unknown scale factor. The corresponding \textit{Euclidean homography} $H \in \R^{3\times 3}$ is 
\begin{equation}
    H := K^{-1} H_{im} K = \gamma \left( R + \frac{\xi \eta^\top}{d}  \right).
\end{equation}  
The scale factor $\gamma$ can be chosen to ensure that $\det(H) = 1$, i.e., to set $H = \pr_{\SL(3)}( R + \xi \eta^\top/d )$. This corresponds to taking $\gamma := \det(R + \xi \eta^\top/d )^{-\frac{1}{3}} = (d/\mathring{d})^{\frac{1}{3}}$ (see \citep{mahony2012nonlinear}).
For the remainder of this paper, all homographies $H$ are taken to be appropriately scaled to satisfy $H \in \SL(3)$.

    We model image maps as real-valued functions defined on the sphere $I: \S^2 \rightarrow [0, 1]$. 
    That is, instead of using pixel coordinates on the image plane, intensity values are defined at points represented by the ray directions $\mathbf{x} \in \S^2$, corresponding to a spherical perspective projection model. 
    We will use this formulation in the remainder of the paper, as working on the sphere adds to the clarity of the mathematical presentation.

    \begin{proposition}[\citep{hua2020nonlinear}] \label{prop1}
		The mapping
        $\bm{\rho}: \SL(3) \times \S^2 \rightarrow \S^2$ 
        given by
		\begin{equation} \label{eq:rho_action}
			\bm{\rho}(H, \mathbf{x}) := \frac{H^{-1} \mathbf{x}}{|H^{-1} \mathbf{x}|},
		\end{equation}
		is a right group action of $\SL(3)$ on $\S^2$. 
	\end{proposition}

    \begin{assumption} \label{assump_const_bright}
    The observed planar surface is Lambertian, i.e., it reflects light uniformly in all directions.
    \end{assumption}
    Assumption \ref{assump_const_bright} corresponds to the classical \textit{brightness constancy} constraint used in direct image alignment methods \citep{irani1999direct}. 
	It postulates that the radiance of a given scene point remains constant across different viewpoints. 
	As a consequence, if $\mathbf{x}$ and $\bm{\rho}_H(\mathbf{x})$ denote corresponding bearing vectors in two images related by a homography action $\bm{\rho}_H$, then the measured intensities along these directions satisfy $I_1(\mathbf{x}) = I_2(\bm{\rho}_H(\mathbf{x}))$.

    \begin{assumption} \label{assump_L2}
    All image maps $I$ are assumed to be square-integrable and twice weakly differentiable in $\L^2(\S^2)$, i.e., they belong to the Sobolev space $W^{2,2}(\S^2)$.
    \end{assumption}

	The assumption of square-integrability of image maps is justified by the fact that real-world images have bounded pixel intensity values and are defined on the compact domain $\S^2$.
	Modelling image maps in the Sobolev space $W^{2,2}(\S^2)$ further accounts for non-smooth images, such as those with edges or intensity discontinuities, by only requiring the existence of weak derivatives rather than strict pointwise differentiability of the image maps.

Under Assumptions \ref{assump_const_bright} and \ref{assump_L2}, the action $\bm{\rho}$ defined in \eqref{eq:rho_action} induces a natural right group action of $\SL(3)$ on image maps in $W^{2,2}(\S^2)$, as established in the following lemma.
\begin{lemma} \label{lemma1}		
The mapping $\bm{\mu}: \SL(3) \times W^{2,2}(\S^2) \to W^{2,2}(\S^2)$, defined by 
	\begin{equation} \label{eq:mu_action}
		\bm{\mu}(H,I) := I \circ \bm{\rho}_{H^{-1}},
	\end{equation}
	is a right group action.
\end{lemma}
\begin{proof}
	Checking the identity property, for any $I \in W^{2,2}(\S^2)$ and $\mathbf{x} \in \S^2$,
	\begin{align*}
		\bm{\mu}(I_3, I)(\mathbf{x}) &= I \circ \bm{\rho}_{I_3}(\mathbf{x}) = I(\mathbf{x}).
	\end{align*}
	Now checking the composition property, for any $I \in W^{2,2}(\S^2)$, any $\mathbf{x} \in \S^2$, and any $G, H \in \mathbf{SL}(3)$, one has
	\begin{align*}
		\bm{\mu}(G, \bm{\mu}(H, I)) &= \bm{\mu}(G, I \circ \bm{\rho}_{H^{-1}}), \\
		&= I \circ \bm{\rho}_{H^{-1}} \circ \bm{\rho}_{G^{-1}}, \\
		&= I \circ \bm{\rho}_{G^{-1} H^{-1}}, \\
		&= I \circ \bm{\rho}_{(HG)^{-1}} = \bm{\mu}(HG, I).
	\end{align*}
Finally, we show that $W^{2,2}(\S^2)$ is closed under the action $\bm{\mu}$, i.e.,
that $\bm{\mu}(H,I) \in W^{2,2}(\S^2)$ for all $I \in W^{2,2}(\S^2)$ and fixed $H \in \SL(3)$.
Recall that $I \in W^{2,2}(\S^2)$ implies $I \in \L^2(\S^2)$ with
$\tilde{\nabla} I, \tilde{\nabla}^2 I \in \L^2(\S^2)$.
We begin by showing that $\bm{\mu}(H,I) \in \L^2(\S^2)$.
Since $H \in \SL(3)$ is fixed, the map $\bm{\rho}_H : \S^2 \to \S^2$ is a smooth diffeomorphism.
By definition, $ \bm{\mu}(H,I)(\mathbf{x}) = I(\bm{\rho}_{H^{-1}}(\mathbf{x}))$.
Using the change of variables $\mathbf{y} = \bm{\rho}_{H^{-1}}(\mathbf{x})$, we obtain
\begin{align*}
\|\bm{\mu}(H,I)\|^2
&= \int_{\S^2} I(\bm{\rho}_{H^{-1}}(\mathbf{x}))^2 \, dV_g(\mathbf{x}) \\
&= \int_{\S^2} I(\mathbf{y})^2 \, |\det(\tD \bm{\rho}_H(\mathbf{y}))| \, dV_g(\mathbf{y}).
\end{align*}
The differential of $\bm{\rho}_H$ is given by
\begin{equation}\label{eq:Drho_recall}
\tD \bm{\rho}_H(\mathbf{y})
= \frac{1}{|H^{-1}\mathbf{y}|}\,\pi_{\bm{\rho}_H(\mathbf{y})} H^{-1}.
\end{equation}
For any orthonormal basis $(v_1,v_2)$ of $T_{\mathbf{y}}\S^2$, the Jacobian determinant satisfies
\[
|\det \tD \bm{\rho}_H(\mathbf{y})|
= |(\tD \bm{\rho}_H(\mathbf{y}) v_1)
\times (\tD \bm{\rho}_H(\mathbf{y}) v_2)|
= \frac{|H^\top \mathbf{y}|}{|H^{-1}\mathbf{y}|^3}.
\]
This expression is continuous and strictly positive on the compact manifold $\S^2$.
Therefore, there exists a constant $M_0(H)>0$ such that $ |\det(\tD \bm{\rho}_H(\mathbf{y}))| \le M_0(H)$, for all $\mathbf{y} \in \S^2$. It follows that
\[
\|\bm{\mu}(H,I)\|^2
\le M_0(H)\|I\|^2 < \infty,
\]
and hence $\bm{\mu}(H,I) \in \L^2(\S^2)$.

We now consider the first covariant derivative.
We have $ \tilde{\nabla} \bm{\mu}(H,I)(\mathbf{x})
= \tD \bm{\rho}_{H^{-1}}(\mathbf{x})^\top
(\tilde{\nabla} I)(\bm{\rho}_{H^{-1}}(\mathbf{x}))$.
Using again the change of variables $\mathbf{y} = \bm{\rho}_{H^{-1}}(\mathbf{x})$, we obtain
\begin{align*}
&\|\tilde{\nabla} \bm{\mu}(H,I)\|^2 \\
&= \int_{\S^2}
\big|
\tD \bm{\rho}_{H^{-1}}(\bm{\rho}_H(\mathbf{y}))^\top
(\tilde{\nabla} I)(\mathbf{y})
\big|^2
|\det(\tD \bm{\rho}_H(\mathbf{y}))|
\, dV_g(\mathbf{y}).
\end{align*}
Since $\bm{\rho}_H$ is a smooth diffeomorphism and $H$ is fixed, both $\tD \bm{\rho}_{H^{-1}}$ and $|\det(\tD \bm{\rho}_H)|$ are continuous and bounded on $\S^2$.
Therefore, there exists a constant $M_1(H)>0$ such that
$$
\|\tilde{\nabla} \bm{\mu}(H,I)\|^2
\le M_1(H)\|\tilde{\nabla} I\|^2 < \infty.
$$
Finally, we consider the second covariant derivative.
A direct computation yields
\begin{align*}
\tilde{\nabla}^2 \bm{\mu}(H,I)(\mathbf x)[u,v] 
&= \tilde{\nabla}_u \big(\tilde{\nabla} \bm{\mu}(H,I) \big)[v] \\
&= \tD \bm{\rho}_{H^{-1}}(\mathbf x)^\top \, (\tilde{\nabla}^2 I)(\bm{\rho}_{H^{-1}}(\mathbf x)) \, \tD \bm{\rho}_{H^{-1}}(\mathbf x)[u,v] \\
&\quad + \sum_{i=1}^2 (\tilde{\nabla} I)(\bm{\rho}_{H^{-1}}(\mathbf x)) \cdot (\tilde{\nabla}_u \tD \bm{\rho}_{H^{-1}}(\mathbf x)[v]) ,
\end{align*}
Since $\bm{\rho}_{H^{-1}}$ is smooth and $\S^2$ is compact,
both $\tD \bm{\rho}_{H^{-1}}$ and its covariant derivative
$\tilde{\nabla} \tD \bm{\rho}_{H^{-1}}$ are bounded.
Consequently, there exists a constant $M_2(H)>0$ such that
\[
\|\tilde{\nabla}^2 \bm{\mu}(H,I)\|^2
\le M_2(H)\big(\|\tilde{\nabla}^2 I\|^2 + \|\tilde{\nabla} I\|^2\big) < \infty.
\]
Since $I \in W^{2,2}(\S^2)$ implies that
$I$, $\tilde{\nabla} I$, and $\tilde{\nabla}^2 I$ all belong to $\L^2(\S^2)$,
we conclude that $\bm{\mu}(H,I) \in W^{2,2}(\S^2)$.
Therefore, $W^{2,2}(\S^2)$ is closed under the action of $\bm{\mu}$.
This completes the proof.

\end{proof}

\begin{remark}
Despite $\bm{\mu}$ being defined on $W^{2,2}(\S^2)$, the output space of interest is the orbit of a fixed image map $I$ under homography transformations, given by
$$ \mathrm{orb}_{\bm{\mu}}(I) := \{  I \circ \bm{\rho}_{H^{-1}} \mid H \in \SL(3)\} \subset W^{2,2}(\S^2). $$
\end{remark}

\section{Application to homography estimation} \label{sec:observability}
Consider a camera that is moving according to some measured linear and angular velocity while capturing images of a planar scene.
Let $H(t) \in \SL(3)$ denote the homography that maps the current image to a fixed reference image, with left-invariant kinematics given by 
\begin{equation} \label{eq:sys}
	\dot{H} = HU, 
\end{equation}
	where $U \in \gothsl(3)$ is the group velocity derived from the camera motion.
	We consider the case where the group velocity is fully measured.
\begin{assumption}
	The group velocity $U \in \gothsl(3)$ is available from measurements.
\end{assumption}

\subsection{Homography Observer Design}

	Let $\hat{H} \in \SL(3)$ denote the estimate of $H$, and define the dynamics of the proposed observer to be
	\begin{align} \label{eq:observer}
		\dot{\hat{H}} &= \hat{H}U + \Delta \hat{H}, &
		\hat{H}(0) &= I_3,
	\end{align}
	where $\Delta \in \gothsl(3)$ denotes the correction term that remains to be designed.
	
	Define the group error $E := \hat{H}H^{-1}$, its dynamics are
	\begin{equation} \label{eq:error}
		\begin{split}
			\dot{E} &=  \hat{H}UH^{-1} + \Delta H\hat{H}^{-1} - \hat{H}UH^{-1} = \Delta E.
		\end{split}
	\end{equation}
	The problem of the observer design is to identify a correction term $\Delta \in \gothsl(3)$ that ensures that the group error $E$ converges to the identity $I_3$, and therefore
	\begin{equation*}
		\hat{H} = \hat{H} H^{-1} H = E H \to I_3 H = H.
	\end{equation*}
In what follows, the correction term $\Delta$ is obtained from the gradient descent direction of a suitable non-degenerate cost function.
This is conceptually similar to the designs based on image feature measurements \citep{hua2019feature,hua2020nonlinear}, with the key difference that here the cost function is constructed directly from raw image intensities.

We define the reference image as a map $\mathring{I} : \mathcal{X} \rightarrow [0,1]$, $\mathring{I} \in W^{2,2}(\mathcal{X})$, where $\mathcal{X} := \mathrm{dom}(\mathring{I}) \subset \S^2$ denotes the fixed reference image domain, corresponding to the projection of a planar surface patch onto the spherical image surface. 
Although $\mathcal{X}$ may represent only a subset of the full image, we refer to $\mathring{I}$ as the \emph{reference image} throughout the estimation problem, since all computations are performed on this fixed domain.
	
Assuming the scene is planar, the current image map can be expressed in terms of the homography $H$ as
	\begin{align*}
		I &:= \bm{\mu}(H, \mathring{I}), \\
		I(\mathbf{x}) &= \bm{\mu}(H, \mathring{I})(\mathbf{x})
		= \mathring{I}(\bm{\rho}(H^{-1},\mathbf{x})).
	\end{align*}
Using a homography estimate $\hat{H}$, we define the \emph{warped image} (or error image) map by exploiting the group action property of $\bm{\mu}$:
	\begin{equation} \label{eq:warped_image}
		I^e := \bm{\mu}(\hat{H}^{-1}, I) = I \circ \bm{\rho}_{\hat{H}} = \mathring{I} \circ \bm{\rho}_{E} = \bm{\mu}(E^{-1}, \mathring{I}).
	\end{equation}
The map $I^e$ therefore provides a direct measure of the estimation error in the image space, such that $I^e = \mathring{I}$ if $E = I_3$ (i.e., $\hat{H} = H$).

\begin{figure}[ht]
    \centering
    \includegraphics[width=0.5\linewidth]{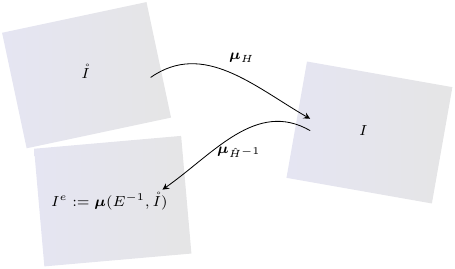}
    \caption{In the image space, the transformation $\bm{\mu}_H$ parameterises image $I$ in terms of the reference image $\mathring{I}$, so the $\bm{\mu}_{\hat{H}^{-1}}$ transformation generates a warped (error) image, which converges to the reference image when $\hat{H} \rightarrow H$.}
    \label{fig:enter-label}
\end{figure}

	Using the estimate $\hat{H}$ and pixel intensity information of the reference and current images, we introduce a \emph{photometric} cost function $\bm{\mathcal{C}} : \SL(3) \times W^{2,2}(\calX) \rightarrow \R_{\geq 0}$ associated with the observer \eqref{eq:observer}, defined as 
	\begin{align}
		\bm{\mathcal{C}}(\hat{H}, I) &:= \frac{1}{2} \Vert I^e - \mathring{I} \Vert^2, \notag \\ 
		&= \frac{1}{2} \int_{\mathcal{X}} \left( \bm{\mu}(\hat{H}^{-1}, I)(\mathbf{x}) - \mathring{I}(\mathbf{x}) \right)^2 d V_g.
		\label{eq:direct_cost_H}
	\end{align}

\begin{remark}
The warped image map \eqref{eq:warped_image} is defined on the time-varying domain 
$\mathrm{dom}(\bm{\mu}(\hat{H}(t)^{-1}, I))$ that depends on the current homography estimate $\hat{H}(t)$.  
To ensure that the cost \eqref{eq:direct_cost} can be evaluated over the entire fixed domain $\calX$, the map $I^e$ is restricted to the intersection
$\mathcal{X}^e := \mathcal{X} \cap \mathrm{dom}\big(\bm{\mu}(\hat{H}(t)^{-1}, I)\big)$,
and extended by zero on $\mathcal{X} \setminus \mathcal{X}^e$. 
Although this extension may introduce a jump discontinuity across the boundary $\partial \mathcal{X}^e$, the cost \eqref{eq:direct_cost} remains well-defined since the image maps are elements of the Sobolev space $W^{2,2}(\mathcal{X})$.
\end{remark}

	Since $I, \mathring{I} \in W^{2,2}(\calX)$, the cost \eqref{eq:direct_cost_H} is well-defined and differentiable in the weak sense.
	Additionally, exploiting the definition of the warped image,
	it can be equivalently expressed in terms of the homography error $E$ as
	\begin{align}
		\bm{\mathcal{C}}(E, \mathring{I}) &= \frac{1}{2} \int_{\mathcal{X}} \left( \bm{\mu}(E^{-1}, \mathring{I})(\mathbf{x}) - \mathring{I}(\mathbf{x}) \right)^2 d V_g.
		\label{eq:direct_cost}
	\end{align}

	More generally, the cost \eqref{eq:direct_cost_H} is \textit{right-invariant} in the sense that  
	$\bm{\calC}(\bm{\phi}(Q,\hat{H}), \bm{\mu}(Q,I)) = \bm{\calC}(\hat{H}, I)$  
	for all $Q \in \SL(3)$.
	This property is particularly convenient because it allows all subsequent computations to be parameterised purely in terms of $E$. Since $\mathring{I}$ is constant, the evolution of $\eqref{eq:direct_cost}$ is  driven fully by the error system \eqref{eq:error}.

	Local convergence of the observer error requires the cost to be non-degenerate. The following lemma provides a necessary condition for $\bm{\calC}$ to satisfy this property.
    \begin{lemma} \label{thm:lemma_non_degeneracy}
		Suppose that the image maps $\mathring{I}$ and $I$ are twice weakly differentiable 
		and let $\tilde{\nabla} \mathring{I}(\mathbf{x})$ denote the reference image map gradient at $\mathbf{x}$. 
		If there exists a measurable subset $\calU \subset \calX$ of strictly positive measure such that
		\begin{align}
			\int_{\calU} \left\langle \tilde{\nabla} \mathring{I}(\mathbf{x}) \mathbf{x}^\top, \Delta \right\rangle^2 dV_g(\mathbf{x}) > 0,
			\label{eq:span_of_gradients}
		\end{align}
		then the cost function \eqref{eq:direct_cost} has an isolated global minimum at $E = I_3$.
	\end{lemma}

	\begin{proof}
		It is straightforward to see that $\bm{\mathcal{C}}(I_3, \mathring{I}) = 0$ and is therefore a global minimum.
		To show that this is an isolated minimum, it suffices to show that the Hessian of $\bm{\mathcal{C}}$ at the identity $I_3$ is non-degenerate.
		The first order derivative of $\bm{\mathcal{C}}$ is computed as follows,
		\begin{align*}
			\tD_1\bm{\mathcal{C}}(E, \mathring{I})[\Delta E] 
			&= \int_{\mathcal{X}} \left(I^e(\mathbf{x}) - \mathring{I}(\mathbf{x})\right) \tD(\mathring{I} \circ \bm{\rho}_{\mathbf{x}})(E)[\Delta E]dV_g \\
			&= \int_{\mathcal{X}} \left. \left(I^e(\mathbf{x}) - \mathring{I}(\mathbf{x})\right) \frac{d}{dt} \mathring{I} \circ \bm{\rho}(\exp(t\Delta)E, {\mathbf{x}}) \right|_{t=0} dV_g \\
			&= \int_{\mathcal{X}} \left. \left(I^e(\mathbf{x}) - \mathring{I}(\mathbf{x})\right) \frac{d}{dt} \mathring{I} \circ \bm{\rho}\left(E, \bm{\rho}(\exp(t\Delta), \mathbf{x})\right) \right|_{t=0} dV_g \\
			&= \int_{\mathcal{X}} \left. \left(I^e(\mathbf{x}) - \mathring{I}(\mathbf{x})\right) \frac{d}{dt} I^e \circ \bm{\rho}(\exp(t\Delta), \mathbf{x}) \right|_{t=0}dV_g \\
			&= \int_{\mathcal{X}} \left(I^e(\mathbf{x}) - \mathring{I}(\mathbf{x})\right) \tilde{\tD} I^e(\mathbf{x}) \tD \bm{\rho}_\mathbf{x}(I_3)[\Delta] dV_g. 
		\end{align*}
		where $\tilde{\tD} I^e(\mathbf{x})$ denotes the transpose of the warped image weak gradient at $\mathbf{x}$; that is, $\tilde{\tD} I^e(\mathbf{x}) = \tilde{\nabla} I^e(\mathbf{x})^\top$. 
		The differential $\tD \bm{\rho}_\mathbf{x}(I_3)[\Delta] $ is given by
		\begin{align} \label{eq:d_rho}
			\tD \bm{\rho}_\mathbf{x}(I_3)[\Delta] &= \left. \frac{d}{dt} \right|_{t=0} \bm{\rho}((\exp(t\Delta), \mathbf{x}), \notag \\
			&= \frac{1}{|\mathbf{x}|^2} \left(-\Delta \mathbf{x} |\mathbf{x}| + \frac{\mathbf{x} \mathbf{x}^\top \Delta \mathbf{x}}{|\mathbf{x}|} \right), \notag \\
			&= \left(\mathbf{x} \mathbf{x}^\top - I_3 \right) \Delta \mathbf{x} = - \pi_{\mathbf{x}} \Delta \mathbf{x},
		\end{align}
		and since $ \tilde{\tD} \mathring{I}(\mathbf{x})\pi_\mathbf{x} = \tilde{\tD} \mathring{I}(\mathbf{x})$, it follows
		\begin{align*}
			\tD_1\bm{\mathcal{C}}(E, \mathring{I})[\Delta E] &= - \int_{\mathcal{X}} \left( I^e(\mathbf{x}) - \mathring{I}(\mathbf{x}) \right) \tilde{\tD} I^e(\mathbf{x}) \Delta \mathbf{x} dV_g.
		\end{align*}
		Now, the second order derivative of $\bm{\mathcal{C}}$ about $I_3$ is
		\begin{align*}
		\tD_1^2\bm{\mathcal{C}}(I_3, \mathring{I})[\Delta, \Delta] 
		    &= - \int_{\mathcal{X}} \left( \tD(\mathring{I} \circ \bm{\rho}_{\mathbf{x}})(E)[\Delta E] \right) \tilde{\tD} I^e(\mathbf{x}) \Delta \mathbf{x} dV_g \\ & \qquad - \left. \int_{\mathcal{X}} \left( I^e(\mathbf{x}) - \mathring{I}(\mathbf{x}) \right) \ddt(\tilde{\tD} I^e(\mathbf{x}) \Delta \mathbf{x}) dV_g \right|_{E=I_3}, \\ 
			&= - \int_{\mathcal{X}} \left( \tilde{\tD} \mathring{I}(\mathbf{x}) \tD \bm{\rho}_\mathbf{x}(I_3)[\Delta] \right) \tilde{\tD} \mathring{I}(\mathbf{x}) \Delta \mathbf{x} dV_g, \\ 
			&= \int_{\mathcal{X}} \left( \tilde{\tD} \mathring{I}(\mathbf{x}) \Delta \mathbf{x} \right)^2 dV_g,
		\end{align*}
		and, hence 
		\begin{align}
			\tD_1^2\bm{\mathcal{C}}(I_3, \mathring{I})[\Delta, \Delta]
			&= \int_{\mathcal{X}} \left( \tr( \tilde{\tD} \mathring{I}(\mathbf{x}) \Delta \mathbf{x}) \right)^2 dV_g, \notag \\
			&= \int_{\mathcal{X}} \left( \tr( \mathbf{x} \tilde{\tD} \mathring{I}(\mathbf{x}) \Delta ) \right)^2 dV_g, \notag \\
			&= \int_{\mathcal{X}} \left\langle \tilde{\nabla} \mathring{I}(\mathbf{x}) \mathbf{x}^\top, \Delta \right\rangle^2 dV_g.
			\label{eq:hessian_equation}
		\end{align}
		Observe that $\text{grad}\mathring{I}(\mathbf{x})\mathbf{x}^\top \in \gothsl(3)$ as $\tr(\text{grad}\mathring{I}(\mathbf{x})\mathbf{x}^\top) = \tr(\tilde{\tD} \mathring{I}(\mathbf{x})\mathbf{x}) = 0$.
		Now suppose the Hessian is degenerate; i.e. there exists a non-zero $\Delta \in \gothsl(3)$ such that $\tD_1^2\bm{\mathcal{C}}(I_3, \mathring{I})[\Delta, \Delta] = 0$.
		Then, since the integrand in \eqref{eq:hessian_equation} is non-negative for every $\mathbf{x}$ and $\tilde{\nabla} \mathring{I}$ exists almost everywhere on $\calX$, it must be that 
		\begin{align} \label{eq:image_obs_condition}
			\left\langle \tilde{\nabla} \mathring{I}(\mathbf{x}) \mathbf{x}^\top, \Delta \right\rangle = 0,
		\end{align}
		for almost every $\mathbf{x} \in \mathcal{X}$.
    	This implies that for any subset $\calU \subseteq \mathcal{X}$ of positive measure,
		\begin{align*}
			\int_{\mathcal{U}} \left\langle \tilde{\nabla} \mathring{I}(\mathbf{x}) \mathbf{x}^\top, \Delta \right\rangle^2 dV_g = 0.
		\end{align*}
		That is, $\langle \tilde{\nabla} \mathring{I}(\mathbf{x}) \mathbf{x}^\top, \Delta \rangle$ vanishes almost everywhere on $\calU$.
		This contradicts the assumption \eqref{eq:span_of_gradients}, and therefore the Hessian \eqref{eq:hessian_equation} must be non-degenerate, and the cost has an isolated global minimum at $E=I_3$.
		
	\end{proof}

	\begin{lemma}
    The non-degeneracy condition \eqref{eq:span_of_gradients} holds if and only if the stabilizer
    \begin{align} \label{eq:stab_condition}
        \mathfrak{stab}_{\bm{\mu}}(\mathring{I}) := \{ E \in \SL(3) \mid \mathring{I} \circ \bm{\rho}_E \simeq \mathring{I} \},
    \end{align} 
    is discrete, where, given $I_1,I_2 \in W^{2,2}(\calX)$, we define $I_1 \simeq I_2$ to mean that $I_1(\mathbf{x}) = I_2(\mathbf{x})$ for almost all $\mathbf{x} \in \calX$.
	Equivalently, the kernel of the differential
	\begin{align*}
		\ker \tD\bm{\mu}_{\mathring{I}}(I_3) := \{ \Delta \in \gothsl(3) \; \mid \; \tD\bm{\mu}_{\mathring{I}}(I_3)[\Delta] \simeq 0 \} = \gothsl(3),
	\end{align*} 
	where we write $\tD\bm{\mu}_{\mathring{I}}(I_3)[\Delta] \simeq 0$ to mean that $\tD\bm{\mu}_{\mathring{I}}(I_3)(\mathbf{x}) = 0$ for almost all $\mathbf{x} \in \calX$.
    \end{lemma}
    
    \begin{proof}
    The relation \eqref{eq:image_obs_condition} can be equivalently expressed as $\langle \tilde{\nabla} \mathring{I}(\mathbf{x}), \pi_{\mathbf{x}}\Delta \mathbf{x} \rangle = 0$, and hence
    \begin{align}\label{eq:kernel_condition}
		\left\langle \tD \bm{\rho}_{\mathbf{x}}(I_3)[\Delta], \tilde{\nabla} \mathring{I}(\mathbf{x}) \right\rangle = 0,
	\end{align}
    for almost every $\mathbf{x} \in \calX$. 
	If this holds for all $\Delta \in \sl(3)$, then it is exactly the kernel condition in the lemma statement.
	
	To see that this implies the stabilizer condition \eqref{eq:stab_condition}, suppose there exists a non-zero $\Delta \in \gothsl(3)$ satisfying \eqref{eq:kernel_condition}. 
	Then $\mathring{I}$ remains invariant almost everywhere along the flow lines generated by $t \mapsto \bm{\rho}(\exp(t\Delta), \mathbf{x})$ for all $t \in \R$.
	Specifically, one has
	\begin{align*}
		& \ddt \frac{1}{2} \int_{\calX} \vert \mathring{I} \circ \bm{\rho}_{\exp(t\Delta)}(\mathbf{x}) - \mathring{I}(\mathbf{x}) \vert^2 dV_g  \\
		&= \int_{\calX} \left\langle \mathring{I} \circ \bm{\rho}_{\exp(t\Delta)}(\mathbf{x}) - \mathring{I}(\mathbf{x}) 
			\; ,
			\tilde{\tD} \mathring{I} (\bm{\rho}_{\exp(t\Delta)}(\mathbf{x})) \tD \bm{\rho}_{\bm{\rho}_{\exp(t\Delta)}(\mathbf{x})}(I_3)[\Delta] \right\rangle d V_g \\
		&= \int_{\calX} \langle \mathring{I} \circ \bm{\rho}_{\exp(t\Delta)}(\mathbf{x}) - \mathring{I}(\mathbf{x}) 
			\; , \;
			0 \rangle d V_g = 0.
	\end{align*}
    It follows that $ \mathring{I} \circ \bm{\rho}_{\exp(t\Delta)}(\mathbf{x}) = \mathring{I}(\mathbf{x})$ for almost all $\mathbf{x}$ for all $t \in \R$.
    Thus, $\exp(t\Delta)$ lies in the stabilizer subgroup $\mathfrak{stab}_{\bm{\mu}}(\mathring{I})$, which is therefore continuous.
	Conversely, this also shows that if $\mathfrak{stab}_{\bm{\mu}}(\mathring{I})$ is continuous, then there exists a nonzero $\Delta \in \mathfrak{stab}_{\bm{\mu}}(\mathring{I})$ in its Lie algebra that satisfies \eqref{eq:image_obs_condition}.
    \end{proof}
    \begin{remark}
	The standard non-degeneracy condition in observer design based on feature measurements (see \cite{mahony2013observers}) requires a finite set of measured features $(\mathring{y}_1, \dots, \mathring{y}_n)$ to satisfy
	$ \bigcap^n_{i=1} \mathfrak{stab}_{\bm{\rho}}(\mathring{y}_i) = \{I_3\}$. 
	Condition \eqref{eq:stab_condition} can be viewed as a continuous counterpart that extends this requirement to intensity-based measurements modelled as maps in the Sobolev space $W^{2,2}(\calX)$.
\end{remark}

\subsection{Generation of degenerate images}
Degeneracies in the cost arise when the reference image admits continuous symmetries captured by its stabilizer subgroup. 
To illustrate situations in which condition \eqref{eq:stab_condition} is violated, we construct reference images with intrinsic continuous symmetries. That is, each reference image $\mathring{I}$ is constructed such that 
$\mathring{I}(\bm{\rho}(\exp(t \Delta), \mathbf{x}))=\mathring{I}(\mathbf{x})$
for all $t \in \R$ and $\mathbf{x} \in \calX$, for some $\Delta \in \gothsl(3)$ generating a one-parameter subgroup $\{\exp(t \Delta)\} \subset \mathfrak{stab}_{\bm{\mu}}(\mathring{I})$.

We numerically synthesise such images by sampling a set of points in $\calX$ and propagating them along their orbits induced by the corresponding group action for multiple values of~$t$.
The procedure is implemented in Python, and we generate six reference images, shown in Fig.~\ref{fig:generated_images}. Each image is invariant under the flow of a specific one-parameter subgroup of $\SL(3)$, generated by an element $\Delta \in \gothsl(3)$ inducing a continuous projective transformation of the image.

\begin{figure*}
\begin{center}
    \includegraphics[width=\textwidth]{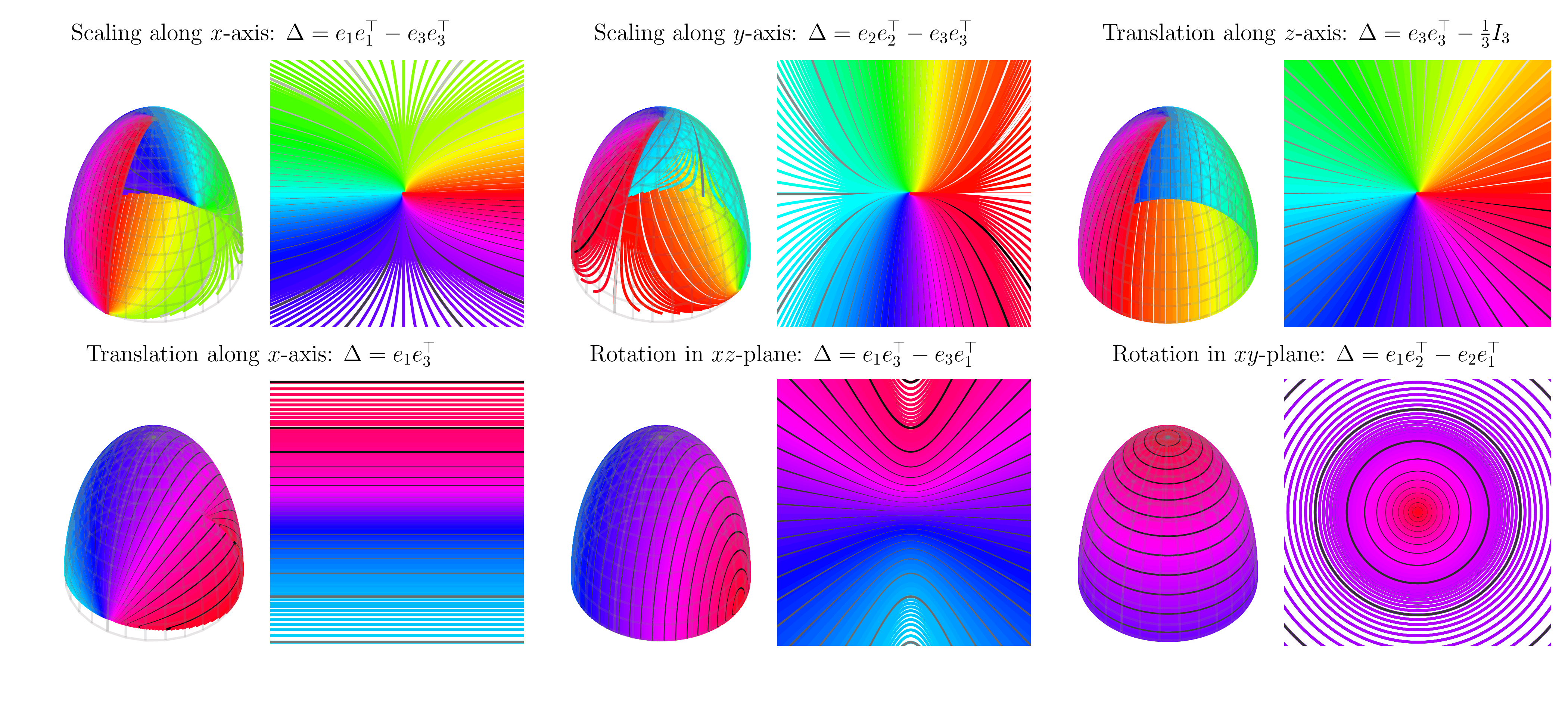}
		\caption{Numerically constructed degenerate reference images exhibiting continuous symmetries invariant under one-parameter subgroups of $\SL(3)$. For each case, a hemisphere representation (left) and the corresponding planar image projection (right) are shown. Image points are grouped into orbits induced by the associated group action; all points in the same orbit share the same color (intensity). Selected orbits are highlighted in grayscale to emphasise the underlying symmetry structure.}
        \label{fig:generated_images}
    \end{center}
\end{figure*}

\section{Gradient-based observer derivation} \label{sec:observer_derivation}

    The cost function defined in \eqref{eq:direct_cost} can be viewed as a candidate Lyapunov function for the observer design, where the correction term can be obtained directly from its gradient flow, as outlined in the following theorem.
	\begin{theorem} \label{thm1}
		Consider the kinematics \eqref{eq:sys} and assume that the group velocity $U \in \gothsl(3)$ is known.
        Consider the observer \eqref{eq:observer} and the cost function $\bm{\calC}$ defined by \eqref{eq:direct_cost} and suppose that the cost is non-degenerate in the sense of Lemma \ref{thm:lemma_non_degeneracy}. Define the correction term 
		\begin{equation} \label{eq:deltaconst}
			\Delta = k_{\Delta}  \int_{\mathcal{X}} r(\mathbf{x})  \tilde{\nabla} I^e(\mathbf{x}) \mathbf{x}^\top dV_g, 
		\end{equation}
		where $r(\mathbf{x}) = I^e(\mathbf{x}) - \mathring{I}(\mathbf{x})$ denotes the image intensity residual at $\mathbf{x}$ and $k_{\Delta} > 0$ a chosen gain.
	Then, its time derivative satisfies $\dot{\bm{\mathcal{C}}}(E) = - |\Delta|^2/k_{\Delta}$, and the equilibrium $E = I_3$ of the autonomous system \eqref{eq:error} is locally exponentially stable. 
	\end{theorem}
	\begin{proof}
		This theorem can be proved using classical Lyapunov theory by choosing the candidate Lyapunov function $ \bm{\mathcal{L}}(E) := \bm{\mathcal{C}}(E, \mathring{I})$.
		Its time derivative is 
		\begin{align*}
			\dot{ \bm{\mathcal{L}}}(E) 
			&=\int_{\mathcal{X}} \left(I^e(\mathbf{x}) - \mathring{I}(\mathbf{x})\right) \frac{d}{dt} \bm{\mu}(E^{-1}, \mathring{I})(\mathbf{x}) dV_g \\
			&= \int_{\mathcal{X}} r(\mathbf{x}) \frac{d}{dt} \mathring{I} \circ \bm{\rho}_{E}(\mathbf{x})dV_g \\
			&= \int_{\mathcal{X}} r(\mathbf{x})  \frac{d}{dt} \mathring{I} \circ \bm{\rho}_{\mathbf{x}}(E)dV_g \\
			&= \int_{\mathcal{X}} r(\mathbf{x}) \tD(\mathring{I} \circ \bm{\rho}_{\mathbf{x}})(E)[\Delta E]dV_g \\
			&= \int_{\mathcal{X}} r(\mathbf{x}) \tilde{\tD} I^e (\mathbf{x}) \tD \bm{\rho}_\mathbf{x}(I_3)[\Delta] dV_g.
		\end{align*}
		Then, using \eqref{eq:d_rho}, we can write
		\begin{align}
            \dot{ \bm{\mathcal{L}}}(E)
			&= - \int_{\mathcal{X}} r(\mathbf{x}) \tilde{\nabla} I^e(\mathbf{x})^\top \pi_{\mathbf{x}} \Delta \mathbf{x} dV_g \notag \\
			&= - \int_{\mathcal{X}} \left\langle r(\mathbf{x}) \tilde{\nabla}I^e(\mathbf{x}),  \Delta \mathbf{x} \right\rangle dV_g \notag \\
			&= - \left\langle  \int_{\mathcal{X}}r(\mathbf{x}) \tilde{\nabla}I^e(\mathbf{x}) \mathbf{x}^\top dV_g , \Delta \right\rangle. \label{eq:costgrad}
		\end{align}
Finally, choosing the correction $\Delta$ as in \eqref{eq:deltaconst} yields
\begin{align} \label{eq:lyapunov_derivative}
    \dot{\bm{\mathcal{L}}}(E) = - \frac{1}{k_{\Delta}} |\Delta|^2 \leq 0. 
\end{align}
Thus, the time derivative $\dot{\bm{\mathcal{L}}}$ is negative semi-definite, and equal to zero when $\Delta = 0$.
By the non-degeneracy assumption of Lemma \ref{thm:lemma_non_degeneracy}, there exists an open neighbourhood of the global minimum $I_3$ in $\SL(3)$ where the cost has compact connected sub-level sets all containing $I_3$ and no other critical point of the cost since its Hessian at this point is positive definite. Let $\mathfrak{B}^E \subseteq \SL(3)$ be the largest such sub-level set of the cost function on which $\bm{\calC}(E, \mathring{I})$ is proper with respect to $E$.
From this, it follows that $E$ is locally bounded.
Since $\bm{\calL}$ is non-increasing as a result of \eqref{eq:lyapunov_derivative}, one ensures that all solutions with initial condition $E(0) \in \mathfrak{B}^E$ remains in $\mathfrak{B}^E$ for all time.
The error dynamics \eqref{eq:error} are autonomous, then by application of LaSalle's invariance principle \citep{khalil2002nonlinear}, we deduce that all solutions of this system with $E(0) \in \mathfrak{B}^E$ converge to the largest invariant set contained in $\mathfrak{I} = \{E \in \mathfrak{B}^E  \,:\, \Delta(E) = 0 \}$. 
And since $\Delta(E) = 0$ (or equivalently $\nabla \bm{\calC}(E, \mathring{I}) = 0$) on $\mathfrak{B}^E$ implies that $E = I_3$, it follows that the largest invariant subset of $\mathfrak{I}$ is $\{I_3\}$ and $E \rightarrow I_3$ for all $E(0) \in \mathfrak{B}^E$.

To conclude local exponential stability of the equilibrium $E = I_3$, we linearise the error dynamics about $I_3$ by writing $E \approx I_3 + \bm{\varepsilon}^\wedge$, with normal coordinates $\bm{\varepsilon} \in \R^8$.
From \eqref{eq:error}, the first-order approximation of the error dynamics is $ \dot{\bm{\varepsilon}} = \Delta^\vee + \calO(|\bm{\varepsilon}|^2)$.
The measurement residual is $\tilde{\bm{y}}(\bm{\varepsilon}) := I^e - \mathring{I} \in \L^2(\calX)$.
A first-order expansion about $\bm{\varepsilon} = 0$ yields
$ \tilde{\bm{y}}(\bm{\varepsilon}) = C \bm{\varepsilon} + \calO(|\bm{\varepsilon}|^2)$,
where $C:\R^8 \rightarrow \L^2(\calX)$ is the Jacobian operator of the warped image map evaluated at $E = I_3$,
with conjugate $C^\star:\L^2(\calX) \rightarrow \R^8$.
Since the cost function satisfies $\bm{\calC}(E,\mathring{I}) = \frac{1}{2}\|\tilde{\bm{y}}(E)\|^2$,
its second-order expansion in normal coordinates is $\bm{\calC}(\bm{\varepsilon}) = \frac{1}{2}\bm{\varepsilon}^\top C^\star C \bm{\varepsilon} + \calO(|\bm{\varepsilon}|^3)$, which shows that $\Hess \bm{\calC}(I_3, \mathring{I}) = C^\star C$.
By the non-degeneracy assumption of Lemma~\ref{thm:lemma_non_degeneracy},
$C^\star C$ is positive definite.
Since the correction $\Delta$ is chosen from the gradient of the cost, it satisfies in local coordinates $\Delta^\vee = -k_\Delta C^\star C \bm{\varepsilon} + \calO(|\bm{\varepsilon}|^2)$.
The linearised error dynamics therefore reduce to
$\dot{\bm{\varepsilon}} = -k_\Delta C^\star C \bm{\varepsilon}$.
Define the Lyapunov function $\bm{\mathcal{L}}_{\mathrm{lin}}(\bm{\varepsilon}) := \frac{1}{2}|\bm{\varepsilon}|^2$. Its time derivative along the linearised dynamics satisfies
$ \dot{\bm{\mathcal{L}}}_{\mathrm{lin}}(\bm{\varepsilon})
= -k_\Delta \bm{\varepsilon}^\top C^\star C \bm{\varepsilon}
\leq -2k_\Delta \lambda_{\min} \bm{\mathcal{L}}_{\mathrm{lin}}(\bm{\varepsilon})$,
where $\lambda_{\min} > 0$ denotes the smallest eigenvalue of $C^\star C$.
It follows that $\bm{\varepsilon} \to 0$ exponentially, and therefore
the equilibrium $E = I_3$ is locally exponentially stable.

    \end{proof}

\subsection{Matrix-gain extension of the gradient observer}
Using a scalar gain in the correction \eqref{eq:deltaconst} ignores the underlying $8$-dimensional geometry of $\gothsl(3)$, and may limit convergence.
In general, the scalar gain can be replaced by any $8 \times 8$ positive-definite matrix, which allows the gradient components to be scaled differently in the canonical coordinates of $\gothsl(3)$. In particular, the gain can be chosen as the inverse of the Hessian of the cost function, which compensates for variations in the sensitivity of the cost along different directions.
This extension is formalised in the following proposition. 
\begin{proposition}
Consider the homography kinematics \eqref{eq:sys} and observer \eqref{eq:observer}.
	Consider the cost function defined by \eqref{eq:direct_cost} and suppose that the cost is non-degenerate in the sense of Lemma \ref{thm:lemma_non_degeneracy}. 
	Define the correction by
\begin{equation} \label{eq:delta_matrix_gain}
    \Delta = \left( \bm{K}  \left(\int_{\mathcal{X}} r(\mathbf{x}) \tilde{\nabla} I^e(\mathbf{x}) \mathbf{x}^\top dV_g\right)^\vee \right)^\wedge,
\end{equation}
where $\bm{K}= k_\Delta (\Hess\bm{\calC}(I_3, \mathring{I}))^{-1}$, where $\Hess \bm{\calC}(I_3, \mathring{I})$ denotes the Hessian matrix of $\bm{\calC}$ evaluated at $I_3$ given by
\begin{equation} \label{eq:hessian_expression}
    \Hess \bm{\calC}(I_3, \mathring{I}) =
\int_{\mathcal{X}}  (\tilde{\nabla}\mathring{I}(\mathbf{x}) \mathbf{x}^\top)^\vee (\tilde{\nabla}\mathring{I}(\mathbf{x}) \mathbf{x}^\top)^{\vee \top} dV_g.
\end{equation}   
Then, the observer dynamics \eqref{eq:observer} achieve local exponential convergence around the equilibrium $E = I_3$ with convergence rate $k_\Delta$.
\end{proposition}
\begin{proof} 
The Hessian of $\bm{\calC}$ at $I_3$ is the symmetric matrix satisfying 
$ \tD_1^2 \bm{\calC}(I_3, \mathring{I})[\Delta,\Delta] = \langle \Delta^\vee, \Hess  \bm{\calC} \, \Delta^\vee \rangle$.
Using \eqref{eq:hessian_equation}, the second-order derivative of $\bm{\calC}$ at $I_3$ writes 
\begin{align*}
    \tD_1^2\bm{\calC}(I_3, \mathring{I})[\Delta, \Delta] &= \int_{\mathcal{X}} \left\langle\tilde{\nabla}\mathring{I}(\mathbf{x}) \mathbf{x}^\top, \Delta \right\rangle^2 dV_g, \\
    &= \int_{\mathcal{X}} \left\langle (\tilde{\nabla}\mathring{I}(\mathbf{x}) \mathbf{x}^\top)^\vee , \Delta^\vee \right\rangle^2 dV_g, \\
    &= \left\langle \Delta^\vee, \int_{\mathcal{X}}   \left((\tilde{\nabla}\mathring{I}(\mathbf{x}) \mathbf{x}^\top)^\vee (\tilde{\nabla}\mathring{I}(\mathbf{x}) \mathbf{x}^\top)^{\vee \top} dV_g\right) \Delta^\vee \right\rangle. 
\end{align*}
Then, direct identification yields expression \eqref{eq:hessian_expression}.
Now, consider the Lyapunov function $ \bm{\mathcal{L}}(E) := \bm{\mathcal{C}}(E, \mathring{I})$.
	Using \eqref{eq:error}, recall that the time derivative of $\bm{\calL}$ is given by 
		\begin{align*}
			\dot{\bm{\mathcal{L}}}(E)
            &= - \left\langle  \int_{\mathcal{X}}r(\mathbf{x}) \tilde{\nabla}I^e(\mathbf{x}) \mathbf{x}^\top dV_g , \Delta \right\rangle, 
		\end{align*}
		From the definition \eqref{eq:delta_matrix_gain} of $\Delta$ and the properties of the vee-wedge operators, we have
		\begin{align*}  
            \dot{\bm{\mathcal{L}}}(E) 
			&= - \left\langle  \left( \bm{K}^{-1} \Delta^\vee \right)^\wedge , \Delta \right\rangle, \\ 
            &= - \left\langle \bm{K}^{-1} \Delta^\vee , \Delta^\vee \right\rangle = - |\Delta^\vee|^2_{\bm{K}^{-1}}.
		\end{align*}    
    The time derivative $\dot{\bm{\calL}}$ is negative semi-definite. 
	Following the same steps as in the proof of Theorem \ref{thm1} and using LaSalle's invariance principle, we show that $E$ converges to the identity for all solutions with initial conditions $H(0) \in \SL(3)$ and $\hat{H}(0) \in \SL(3)$ such that $E(0) \in \mathfrak{B}^E$. 

To analyse the local convergence properties, consider the first-order approximation $E \approx I_3 + \bm{\varepsilon}^\wedge$. 
Assuming that condition \eqref{eq:stab_condition} holds and substituting the gain $\bm{K} = k_\Delta(C^\star C)^{-1}$, the linearised error dynamics simplify to $\dot{\bm{\varepsilon}} = - k_\Delta (C^\star C)^{-1} C^\star C \bm{\varepsilon} = - k_\Delta \bm{\varepsilon} $.
Computing the time derivative of the Lyapunov function $\bm{\mathcal{L}}_{\mathrm{lin}}(\bm{\varepsilon}) := \frac{1}{2}|\bm{\varepsilon}|^2$ we get $ \dot{\bm{\mathcal{L}}}_{\mathrm{lin}}(\bm{\varepsilon}) = - k_\Delta|\bm{\varepsilon}|^2 = - 2 k_\Delta \bm{\mathcal{L}}_{\mathrm{lin}}(\bm{\varepsilon})$.
It follows that $\bm{\mathcal{L}}_{\mathrm{lin}}(\bm{\varepsilon})$ decreases exponentially and the error $\bm{\varepsilon}$ converges exponentially to zero at rate $k_\Delta$.
\end{proof}

\begin{remark} \label{remark:orthogonal_decomposition}
Another specialisation of the gain $\bm{K} \in \mathbb{S}_+(8)$ exploits the decomposition of the $\mathfrak{sl}(3)$ algebra into orthogonal symmetric and skew-symmetric subspaces. 
Let $\bm{M} := \int_{\mathcal{X}} r(\mathbf{x}) \tilde{\nabla} I^e(\mathbf{x}) \mathbf{x}^\top d\mathbf{x} \in \mathfrak{sl}(3)$ and define the correction 
\begin{equation} \label{eq:delta_orthogonal}
	\Delta = k_s \mathbb{P}_{\mathrm{s}}(\bm{M})+ k_a \mathbb{P}_{\mathrm{a}}(\bm{M}), \quad k_s,k_a > 0, \end{equation}
where the skew-symmetric subspace is the $\so(3)$ algebra, and its orthogonal complement, the symmetric subspace of $\gothsl(3)$, spans infinitesimal projective distortions (i.e., shearing and scaling). 
By tuning $k_s$ and $k_a$, the observer can scale corrections along these subspaces and
outperform \eqref{eq:deltaconst} when homography errors are predominantly due to either projective deformations or rotations. 
Therefore, correction \eqref{eq:delta_orthogonal} offers an effective trade-off between computational efficiency and convergence performance.  

\end{remark}

\section{Simulation results} \label{sec:simulation}

In this section, we present the simulation experiment conducted to verify the performance and convergence of the proposed observer \eqref{eq:observer}. 
From a reference image of a resolution of $256 \times 254$ pixels, the image at a given time was generated using $I = \bm{\mu}(H, \mathring{I})$.
The true homography dynamics are defined according to \eqref{eq:sys}, with initial condition and velocity
\begin{align*}
    H(0) = \begin{pmatrix} 1.031 & 0.051 &  0.087 \\
	-0.051 & 1.031 & -0.144 \\
	0 & 0 &  0.939 \end{pmatrix}, 
	U(t) = \begin{pmatrix} 0 & 0 &  -0.1 \\
	0 & 0 & 0.1 \\
	0 & 0 &  0 \end{pmatrix},
\end{align*}
corresponding to a constant translational motion parallel to the plane, with angular velocity $\Omega(t) = 0$ and $\ddt (\xi/d) = (-0.1, 0.1, 0)^\top$.

The observer dynamics are defined according to \eqref{eq:observer} with initial condition 
$\hat{H}(0) = I_3 $ 
and gain $k_\Delta = 0.01$. 
The true and estimated dynamics are integrated using Euler integration for 3~s with a time-step of 0.01~s using the matrix exponential to ensure that they remain in the Lie group $\SL(3)$ for all time.
Due to discretisation errors, the numerical estimate may drift from $\SL(3)$ and is therefore re-projected onto the group using \eqref{eq:SL3_proj}.

\subsection{Results and discussion}
The results are presented in Figures \ref{fig:sym_skew_errors} and \ref{fig:diff_image}. 
Figure \ref{fig:sym_skew_errors} reports the homography estimation error, $\epsilon_H = | I_3 - E |^2$,  and the normalized image intensity error between the reference and warped images $\epsilon_I = \tfrac{1}{N}\| I^e - \mathring{I} \|^2 = \frac{2}{N}\calC(E, \mathring{I})$, where $N$ is the number of pixels. 
The performance of the scalar-gain correction \eqref{eq:deltaconst} is compared with that obtained using the inverse Hessian–based correction term \eqref{eq:delta_matrix_gain} and with the correction derived from the orthogonal decomposition described in Remark~\ref{remark:orthogonal_decomposition} with $k_s = k_{\Delta}$ and $k_a = 2k_{\Delta}$.

The results indicate that all three observers achieve convergence of both \(\epsilon_H\) and \(\epsilon_I\) to zero.
The observer based on the orthogonal decomposition converges faster than the scalar-gain observer, while the inverse-Hessian-based observer achieves the fastest convergence, with the error decreasing exponentially.
Note that the errors do not converge exactly to zero due to pixel mismatch at the image borders introduced by the warping function, which prevents perfect alignment of all pixels of the warped and reference images even when the observer has converged.

Figure \ref{fig:diff_image} shows the evolution of the warped image and the corresponding difference image with respect to the reference at $t=0\ \text{s}$, $t=0.15\ \text{s}$, and $t=1\ \text{s}$. The results confirm that the warped image progressively converges to the reference image over time.
The intensity-based observer~\eqref{eq:observer} progressively aligns the warped image with the reference. The correction is driven by intensity differences within overlapping regions between $I^e$ and $\mathring{I}$. As the warped image converges toward the reference, the overlap error decreases, the magnitude of the correction diminishes, and the cost converges to zero.

\begin{figure}[ht]
\centering
\scriptsize{Homography estimation error }\\
\includegraphics[scale=.55]{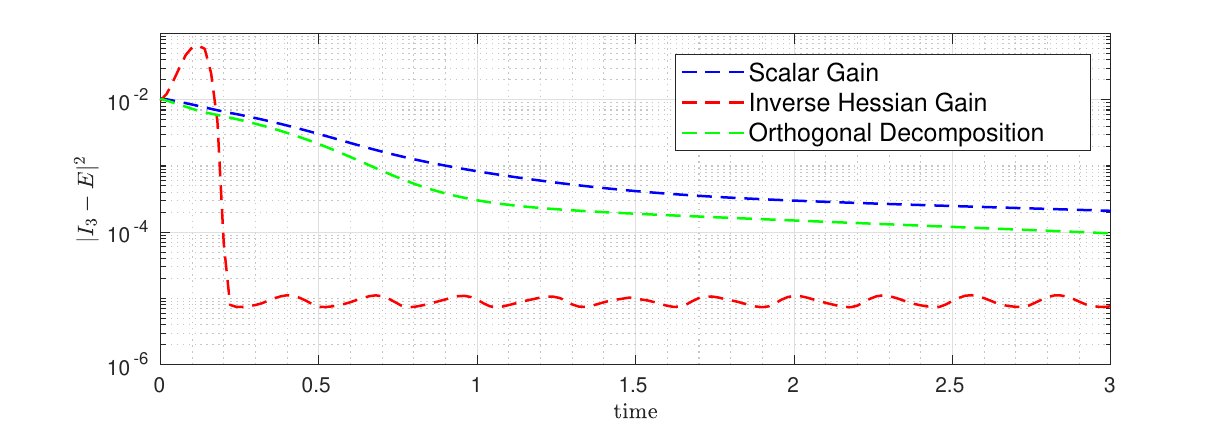}  \\
\scriptsize{Image intensity error }\\ 
\includegraphics[scale=.55]{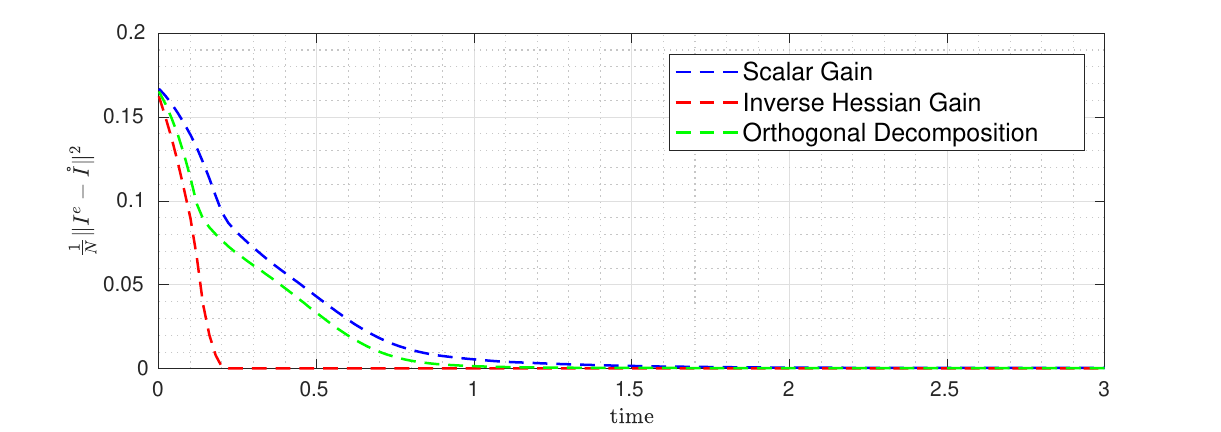}
\caption{Time evolution of the homography estimation error (in semilogarithmic scale) and normalised image intensity error for the observer \eqref{eq:observer} using the single scalar gain (blue), the inverse Hessian gain (red) and the orthogonal decomposition dual-gain (green).}
\label{fig:sym_skew_errors}
\end{figure}

\begin{figure}[ht]
\centering
\includegraphics[width=8.5cm,height=2.5cm]{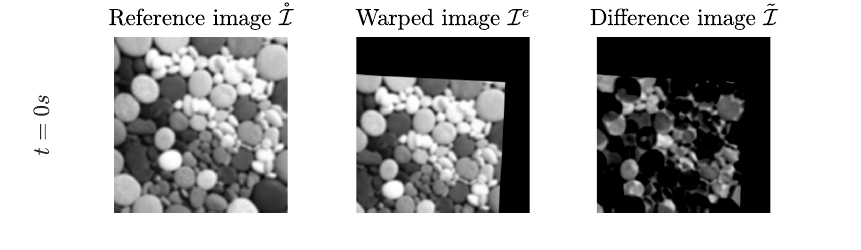} \\[-.1cm]
\includegraphics[width=8.5cm,height=2.5cm]{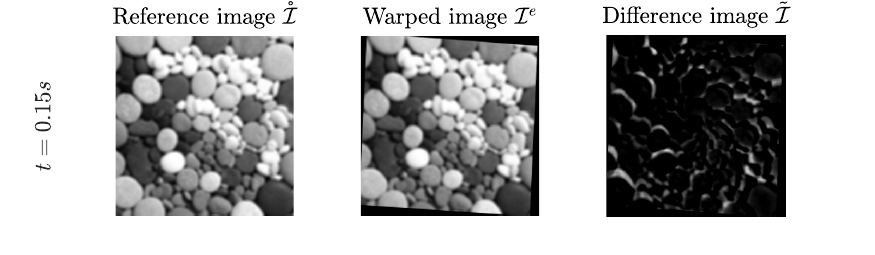} \\[-.1cm]
\includegraphics[width=8.5cm,height=2.5cm]{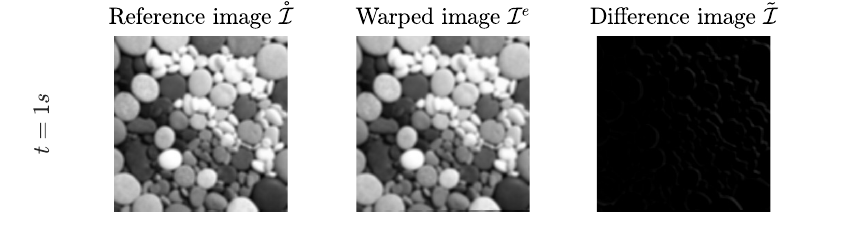} 
\caption{Reference image $\mathring{I}$, warped image $I^e$ and difference image $\tilde{I}$ at selected time intervals, using the observer~\eqref{eq:observer}.
As the observer converges, the warped image progressively aligns with the reference image, and correspondingly, the difference image approaches zero.
}
\label{fig:diff_image}
\end{figure}

\section{Conclusions} \label{sec:conclusion}
This work introduced a nonlinear observer framework for dynamic homography estimation from raw image intensities. By exploiting the $\SL(3)$ Lie group structure and its induced action on the Sobolev space $W^{2,2}(\S^2)$, we defined a photometric cost function and derived a gradient-based observer. Local convergence was further improved through a matrix-gain formulation, in particular by incorporating the Hessian matrix.
A detailed observability analysis established a local non-degeneracy condition and identified degenerate image configurations associated with invariance to $\SL(3)$ subgroups.
Simulation experiments on real image sequences confirmed convergence, robustness, and performance improvements of the proposed observers. This is the first direct, intensity-based nonlinear observer on $\SL(3)$.
The proposed dense framework could be in principle extended to other state estimation problems where the system evolves on a matrix Lie group and dense sensor measurements are available. A natural direction is the estimation of rigid-body motion on the Special Euclidean Group $\SE(2)$ or $\SE(3)$ using depth measurements. 


    \section*{Acknowledgment}
    This work was supported by the Grands Fonds Marins Project Deep-C, the ASTRID ANR project ASCAR (ANR-23-ASTR-0016), the European Union through the Horizon Europe Research and Innovation Programme (Grant Agreement No. 101154194, MEW), and the Australian Research Council through Discovery Grant DP250100112, “Seeing through Space and Time: Spatio-Temporal Event Processing for Robots”.


\bibliographystyle{unsrtnat}
\bibliography{autosam}

\end{document}

%% file: notation_defs.tex
\providecommand{\R}{\mathbb{R}}

\providecommand{\S}{\mathbb{S}}

\providecommand{\SO}{\mathbf{SO}}
\providecommand{\SL}{\mathbf{SL}}
\providecommand{\GL}{\mathbf{GL}}
\providecommand{\SE}{\mathbf{SE}}

\providecommand{\grpG}{\mathbf{G}}

\providecommand{\gothsl}{\mathfrak{sl}}

\providecommand{\gothg}{\mathfrak{g}}

\providecommand{\so}{\mathfrak{so}}

\providecommand{\calC}{\mathcal{C}}

\providecommand{\calL}{\mathcal{L}}
\providecommand{\calM}{\mathcal{M}}
\providecommand{\calN}{\mathcal{N}}
\providecommand{\calO}{\mathcal{O}}

\providecommand{\calU}{\mathcal{U}}

\providecommand{\calX}{\mathcal{X}}

\DeclareMathOperator{\tr}{tr}

\providecommand{\pr}{\mathbb{P}} 

\providecommand{\tT}{\mathrm{T}} 
\providecommand{\td}{\mathrm{d}}
\providecommand{\tD}{\mathrm{D}}

\providecommand{\ddt}{\frac{\td}{\td t}}

\providecommand{\Hess}{\mathrm{Hess}}

\usepackage{accents}
\makeatletter
\providecommand{\scirc}{%
    \hbox{\fontfamily{\rmdefault}\fontsize{0.4\dimexpr(\f@size pt)}{0}\selectfont{\raisebox{-0.52ex}[0ex][-0.52ex]{$\circ$}}}}

\makeatother

\mathchardef\mhyphen="2D



%% file: autosam.bib
@book{hartley2003multiple,
  title={Multiple view geometry in computer vision},
  author={Hartley, Richard and Zisserman, Andrew},
  year={2003},
  publisher={Cambridge university press}
}

@book{ma2004invitation,
  title={An invitation to 3-d vision: from images to geometric models},
  author={Ma, Yi and Soatto, Stefano and Ko{\v{s}}eck{\'a}, Jana and Sastry, Shankar},
  volume={26},
  year={2004},
  publisher={Springer}
}

@inproceedings{geiger2011stereoscan,
  title={Stereoscan: Dense 3d reconstruction in real-time},
  author={Geiger, Andreas and Ziegler, Julius and Stiller, Christoph},
  booktitle={2011 IEEE intelligent vehicles symposium (IV)},
  pages={963--968},
  year={2011},
  organization={Ieee}
}

@inproceedings{forster2014svo,
  title={SVO: Fast semi-direct monocular visual odometry},
  author={Forster, Christian and Pizzoli, Matia and Scaramuzza, Davide},
  booktitle={2014 IEEE international conference on robotics and automation (ICRA)},
  pages={15--22},
  year={2014},
  organization={IEEE}
}

@inproceedings{kerl2013robust,
  title={Robust odometry estimation for RGB-D cameras},
  author={Kerl, Christian and Sturm, J{\"u}rgen and Cremers, Daniel},
  booktitle={2013 IEEE international conference on robotics and automation},
  pages={3748--3754},
  year={2013},
  organization={IEEE}
}

@article{benhimane2007homography,
  title={Homography-based 2d visual tracking and servoing},
  author={Benhimane, Selim and Malis, Ezio},
  journal={The International Journal of Robotics Research},
  volume={26},
  number={7},
  pages={661--676},
  year={2007},
  publisher={Sage Publications Sage UK: London, England}
}

@article{mahony2012nonlinear,
  title={Nonlinear complementary filters on the special linear group},
  author={Mahony, Robert and Hamel, Tarek and Morin, Pascal and Malis, Ezio},
  journal={International Journal of Control},
  volume={85},
  number={10},
  pages={1557--1573},
  year={2012},
  publisher={Taylor \& Francis}
}

@inproceedings{hamel2011homography,
  title={Homography estimation on the special linear group based on direct point correspondence},
  author={Hamel, Tarek and Mahony, Robert and Trumpf, Jochen and Morin, Pascal and Hua, Minh-Duc},
  booktitle={2011 50th IEEE Conference on Decision and Control and European Control Conference},
  pages={7902--7908},
  year={2011},
  organization={IEEE}
}

@inproceedings{szeliski1995direct,
  title={Direct methods for visual scene reconstruction},
  author={Szeliski, Richard and Kang, Sing Bing},
  booktitle={Proceedings IEEE Workshop on Representation of Visual Scenes (In Conjunction with ICCV'95)},
  pages={26--33},
  year={1995},
  organization={IEEE}
}

@article{shum2002construction,
  title={Construction of panoramic image mosaics with global and local alignment},
  author={Shum, Heung-Yeung and Szeliski, Richard},
  journal={International Journal of Computer Vision},
  volume={48},
  number={2},
  pages={151--152},
  year={2002},
  publisher={Springer Netherlands}
}

@article{baker2004lucas,
  title={Lucas-kanade 20 years on: A unifying framework},
  author={Baker, Simon and Matthews, Iain},
  journal={International journal of computer vision},
  volume={56},
  number={3},
  pages={221--255},
  year={2004},
  publisher={Springer}
}

@inproceedings{benhimane2004real,
  title={Real-time image-based tracking of planes using efficient second-order minimization},
  author={Benhimane, Selim and Malis, Ezio},
  booktitle={2004 IEEE/RSJ International Conference on Intelligent Robots and Systems (IROS)(IEEE Cat. No. 04CH37566)},
  volume={1},
  pages={943--948},
  year={2004},
  organization={IEEE}
}

@article{kaminski2004multiple,
  title={Multiple view geometry of general algebraic curves},
  author={Kaminski, Jeremy Yermiyahou and Shashua, Amnon},
  journal={International Journal of Computer Vision},
  volume={56},
  number={3},
  pages={195--219},
  year={2004},
  publisher={Springer}
}

@article{agarwal2005survey,
  title={A survey of planar homography estimation techniques},
  author={Agarwal, Anubhav and Jawahar, CV and Narayanan, PJ},
  journal={Centre for Visual Information Technology, Tech. Rep. IIIT/TR/2005/12},
  year={2005},
  publisher={International Institute of Information Technology Hyderabad, India}
}

@inproceedings{lucas1981iterative,
  title={An iterative image registration technique with an application to stereo vision},
  author={Lucas, Bruce D and Kanade, Takeo},
  booktitle={IJCAI'81: 7th international joint conference on Artificial intelligence},
  volume={2},
  pages={674--679},
  year={1981}
}

@article{szeliski2007image,
  title={Image alignment and stitching: A tutorial},
  author={Szeliski, Richard and others},
  journal={Foundations and Trends{\textregistered} in Computer Graphics and Vision},
  volume={2},
  number={1},
  pages={1--104},
  year={2007},
  publisher={Now Publishers, Inc.}
}

@inproceedings{irani1999direct,
  title={About direct methods},
  author={Irani, Michal and Anandan, Prabu},
  booktitle={International Workshop on Vision Algorithms},
  pages={267--277},
  year={1999},
  organization={Springer}
}

@inproceedings{shum1998construction,
  title={Construction and refinement of panoramic mosaics with global and local alignment},
  author={Shum, Heung-Yeung and Szeliski, Richard},
  booktitle={Sixth International Conference on Computer Vision (IEEE Cat. No. 98CH36271)},
  pages={953--956},
  year={1998},
  organization={IEEE}
}

@book{hebey1996sobolev,
  title={Sobolev spaces on Riemannian manifolds},
  author={Hebey, Emmanuel},
  volume={1635},
  year={1996},
  publisher={Springer Science \& Business Media}
}

@article{chan2024meyers,
  title={The Meyers-Serrin theorem on Riemannian manifolds: a survey},
  author={Chan, Chi Hin and Czubak, Magdalena},
  journal={arXiv preprint arXiv:2405.13322},
  year={2024}
}

@book{baker2003matrix,
  title={Matrix groups: An introduction to Lie group theory},
  author={Baker, Andrew},
  year={2003},
  publisher={Springer Science \& Business Media}
}

@article{mahony2013observers,
  title={Observers for kinematic systems with symmetry},
  author={Mahony, Robert and Trumpf, Jochen and Hamel, Tarek},
  journal={IFAC Proceedings Volumes},
  volume={46},
  number={23},
  pages={617--633},
  year={2013},
  publisher={Elsevier}
}

@article{hua2019feature,
  title={Feature-based recursive observer design for homography estimation and its application to image stabilization},
  author={Hua, Minh-Duc and Trumpf, Jochen and Hamel, Tarek and Mahony, Robert and Morin, Pascal},
  journal={Asian Journal of Control},
  volume={21},
  number={4},
  pages={1443--1458},
  year={2019},
  publisher={Wiley Online Library}
}

@article{hua2020nonlinear,
  title={Nonlinear observer design on SL (3) for homography estimation by exploiting point and line correspondences with application to image stabilization},
  author={Hua, Minh-Duc and Trumpf, Jochen and Hamel, Tarek and Mahony, Robert and Morin, Pascal},
  journal={Automatica},
  volume={115},
  pages={108858},
  year={2020},
  publisher={Elsevier}
}

@inproceedings{hua2017explicit,
  title={Explicit complementary observer design on special linear group sl (3) for homography estimation using conic correspondences},
  author={Hua, Minh-Duc and Hamel, Tarek and Mahony, Robert and Allibert, Guillaume},
  booktitle={2017 IEEE 56th Annual Conference on Decision and Control (CDC)},
  pages={2434--2441},
  year={2017},
  organization={IEEE}
}

@article{bernal2023bayesian,
  title={Bayesian Filtering for Homography Estimation},
  author={Bernal, Arturo Del Castillo and Decoste, Philippe and Forbes, James Richard},
  journal={IEEE Robotics and Automation Letters},
  year={2023},
  publisher={IEEE}
}

@inproceedings{bouazza2023equivariant,
  title={Equivariant filter for feature-based homography estimation for general camera motion},
  author={Bouazza, Tarek and Ashton, Katrina and van Goor, Pieter and Mahony, Robert and Hamel, Tarek},
  booktitle={2023 62nd IEEE Conference on Decision and Control (CDC)},
  pages={8463--8470},
  year={2023},
  organization={IEEE}
}

@article{bouazza2023nonlinear,
  title={Nonlinear constructive observer design for direct homography estimation},
  author={Bouazza, Tarek and Van Goor, Pieter and Mahony, Robert and Hamel, Tarek},
  journal={IFAC-PapersOnLine},
  volume={56},
  number={2},
  pages={1655--1660},
  year={2023},
  publisher={Elsevier}
}

@book{khalil2002nonlinear,
  title={Nonlinear systems},
  author={Khalil, Hassan K and Grizzle, Jessy W},
  volume={3},
  year={2002},
  publisher={Prentice hall Upper Saddle River, NJ}
}
